\documentclass[11pt]{article}
\usepackage{color,amsmath, amsthm, amsfonts, amssymb, dsfont, graphicx, listings, graphics, epsfig, enumerate, bbm,dsfont,mathrsfs}

\usepackage[margin=1in]{geometry}
\usepackage[fleqn,tbtags]{mathtools}
\usepackage[english]{babel}

\usepackage[round]{natbib} 

\title{Portfolio optimization with two quasiconvex risk measures}
\author{\c{C}a\u{g}{\i}n Ararat\thanks{Bilkent University, Department of Industrial Engineering, Ankara, Turkey, cararat@bilkent.edu.tr.}
}
\date{December 11, 2020}

\addto\captionsenglish {}
\makeatletter \renewenvironment{proof}[1][\proofname] {\par\pushQED{\qed}\normalfont\topsep6\p@\@plus6\p@\relax\trivlist\item[\hskip\labelsep\bfseries#1\@addpunct{.}]\ignorespaces}{\popQED\endtrivlist\@endpefalse} \makeatother
\newtheorem{theorem}{Theorem}[section]

\newtheorem{lemma}[theorem]{Lemma}
\newtheorem{proposition}[theorem]{Proposition}
\newtheorem{assumption}[theorem]{Assumption}
\newtheorem{definition}[theorem]{Definition}
\theoremstyle{definition}
\newtheorem{example}[theorem]{Example}
\newtheorem{remark}[theorem]{Remark}

\numberwithin{equation}{section}
\newcommand{\R}{\mathbb{R}}
\newcommand{\sm}{\!\setminus\!}

\DeclareMathOperator{\dom}{dom}

\let\abs=\envert

\newcommand{\W}{\mathcal{W}}
\renewcommand{\O}{\Omega}

\newcommand{\F}{\mathcal{F}}
\newcommand{\D}{\mathscr{D}}

\newcommand{\X}{\mathcal{X}}
\newcommand{\Y}{\mathcal{Y}}

\newcommand{\A}{\mathscr{A}}

\renewcommand{\P}{\mathscr{P}}
\newcommand{\E}{\mathbb{E}}

\newcommand{\N}{\mathcal{N}}
\renewcommand{\Pr}{\mathbb{P}}

\renewcommand{\a}{\alpha}
\renewcommand{\b}{\beta}
\newcommand{\g}{\gamma}

\newcommand{\1}{\mathbf{1}}

\newcommand{\of}[1]{\ensuremath{\left( #1 \right)}}
\newcommand{\cb}[1]{\ensuremath{ \left\{ #1 \right\} }}
\newcommand{\sqb}[1]{\ensuremath{ \left[ #1 \right] }}
\newcommand{\norm}[1]{\ensuremath{ \left\Vert #1 \right\Vert }}
\newcommand{\ip}[1]{\ensuremath{ \left\langle #1 \right\rangle }}

\def\prehp(#1,#2){\ensuremath{  #1 \cdot #2 }}

\begin{document}
\maketitle
\thispagestyle{empty}

\begin{abstract}
We study a static portfolio optimization problem with two risk measures: a principle risk measure in the objective function and a secondary risk measure whose value is controlled in the constraints. This problem is of interest when it is necessary to consider the risk preferences of two parties, such as a portfolio manager and a regulator, at the same time. A special case of this problem where the risk measures are assumed to be coherent (positively homogeneous) is studied recently in a joint work of the author. The present paper extends the analysis to a more general setting by assuming that the two risk measures are only quasiconvex. First, we study the case where the principal risk measure is convex. We introduce a dual problem, show that there is zero duality gap between the portfolio optimization problem and the dual problem, and finally identify a condition under which the Lagrange multiplier associated to the dual problem at optimality gives an optimal portfolio. Next, we study the general case without the convexity assumption and show that an approximately optimal solution with prescribed optimality gap can be achieved by using the well-known bisection algorithm combined with a duality result that we prove.\\
\\[-5pt]
\textbf{Keywords and phrases:} portfolio optimization, quasiconvex risk measure, minimal penalty function, maximal risk function, Lagrange duality, bisection method\\
\\[-5pt]
\textbf{Mathematics Subject Classification (2010): }90C11, 90C20, 90C90, 91B30, 91G10.
\end{abstract}

\section{Introduction}\label{intro}

Risk measures are functionals that are defined on a linear space of real-valued random variables hosted by a common probability space. In the context of financial mathematics, each random variable can denote the uncertain future worth of an investor's position, and a risk measure assigns a (deterministic) extended real number to the random variable; this number quantifies the initial capital that is needed for compensating the risk of the position.

Introduced in \cite{artzner}, risk measures have been studied extensively in the financial mathematics literature over the last two decades. In \cite{artzner}, the so-called \emph{coherent risk measures} are studied within an axiomatic framework; we provide mathematical formulations of these axioms in Section \ref{setup} for the convenience of the reader. Among the properties of a coherent risk measure, \emph{positive homogeneity} imposes that the risk of a financial position is scalable by the size of the position. While some classical risk measures such as negative expected value and average value-at-risk (see \citet[Example~4.40]{fs:sf}) enjoy this property, positive homogeneity can be found restrictive from a financial point of view. To this end, \emph{convex risk measures} provide a richer class of risk measures where positive homogeneity is not taken for granted. A classical example of a convex but not coherent risk measure is the entropic risk measure, which has a simple expression of the log-sum-exp form (see Example~\ref{entropic}). The reader is referred to \citet[Chapter~4]{fs:sf} for a detailed discussion on convex risk measures.

The convexity property of a risk measure is often motivated by the statement ``Diversification does not increase risk," which emphasizes the role of allocating one's capital into a variety of investment opportunities. More recently, it has been argued that quasiconvexity can be used as a relaxed alternative for convexity as it still captures the idea behind diversification. Therefore, \emph{quasiconvex risk measures}, as argued in \cite{marinacci} and \cite{drapeaukupper}, cover a wider range of functionals that can be used for risk measurement purposes; these include certainty equivalents (see \citet[Example~8]{drapeaukupper} and Section~\ref{certaintyequivalent}) and economic indices of riskiness (see \citet[Example~3]{drapeaukupper}) in addition to convex (and coherent) risk measures described above.

A rich class of problems where risk measures appear naturally is that of portfolio optimization problems. In these problems, one wishes to minimize or control the risk of the future value of a portfolio that consists of multiple risky assets. For special families of asset return distributions, the works \cite{landsmanold,landsman,owadally} study static portfolio optimization problems with a single coherent risk measure that appears in the objective function. More recently, in the previous joint work \cite{deniz} of the author, a static portfolio optimization problem with two coherent risk measures is formulated. In this problem, the decision-maker aims to minimize the value of a principle risk measure, e.g., the risk measure of the portfolio manager, while keeping the value of a secondary risk measure, e.g., the risk measure declared by a regulatory authority, below a critical threshold. In \cite{deniz}, a complete analysis of this problem is provided for the general case of arbitrary asset return distributions and arbitrary coherent risk measures that satisfy certain regularity conditions.

On the other hand, the use of quasiconvex risk measures in portfolio optimization is relatively new. In \cite{elisa}, a static portfolio optimization problem is studied, where the objective function is the composition of a quasiconvex risk measure and a concave functional that is defined on the space of portfolios, and this composition is to be minimized over a convex compact set of portfolios. The main result \citet[Theorem~4]{elisa} provides a sufficient condition for a portfolio to be optimal in terms of a set relation between normal cones and generalized subdifferentials for quasiconvex functions. In particular, the derivations rely on the dual representations for quasiconvex risk measures developed in \cite{drapeaukupper} as well as the general duality theory for quasiconvex functions initiated earlier in \cite{penotvolle}. In \cite{sigrid}, quasiconvex risk measures are used in a dynamic portfolio optimization problem in continuous time in order to model ambiguity-averse preferences. The work \cite{sigrid} also makes use of the dual representation results of \cite{drapeaukupper}.

The aim of the present paper is to extend the static portfolio optimization problem in \cite{deniz} by assuming that both the principle and the secondary risk measures are quasiconvex. In particular, we cover the case where the two risk measures are convex. It should be noted that the extension from the coherent case to the quasiconvex case requires entirely different duality arguments, explaining the mathematical originality of the present paper. On the other hand, compared to the portfolio optimization problem in \cite{elisa} with a single quasiconvex risk measure and a general convex set constraint, our problem assumes that the constraint has a special structure induced by the secondary risk measure. This structure makes it possible to formulate a more explicit dual problem with a linear inequality constraint.

The rest of the paper is organized as follows. After reviewing some basic notions about quasiconvex risk measures in Section~\ref{setup} and introducing the primal problem in Section~\ref{arbitrary}, we break down the analysis of the problem into two steps. First, in Section~\ref{convexcase}, we work under the assumption that the principle risk measure is a convex functional (not necessarily translative though). We formulate the dual problem in Section~\ref{sect1} and prove that (Theorem~\ref{thm1}) there is zero duality gap between the primal and dual problems. In Section~\ref{sect2}, we impose further structural properties on the principle risk measure and prove that (Theorem~\ref{thm2}) a Lagrange multiplier attached to a linear inequality constraint of the dual problem at optimality yields an optimal portfolio vector for the primal problem. Next, in Section~\ref{quasiconvexcase}, we remove the convexity assumption and reformulate the quasiconvex portfolio optimization problem via a family of convex feasibility problems parametrized by a decision variable of the quasiconvex problem. Similar to the results of Section~\ref{convexcase}, we provide a duality-based method to solve each of these feasibility problems. Then, we employ the well-known bisection method that iterates through different values of the parameter of the feasibility problems and stops with prescribed suboptimality in finitely many iterations. Hence, combining the duality result with the bisection method provides a way to find an approximately optimal solution for the portfolio optimization problem under quasiconvex risk measures. Finally, in Section~\ref{examples}, we consider convex risk measures and certainty equivalents in order to illustrate the use of the dual problem, and we discuss the validity of some technical assumptions stated in Section~\ref{convexcase} and Section~\ref{quasiconvexcase}.

\section{Quasiconvex risk measures}\label{setup}

In this section, we fix the notation for the rest of the paper and review some preliminary notions related to risk measures. For the latter, we focus on the more recent quasiconvex framework studied in \cite{marinacci, drapeaukupper}.

Let $n\in\mathbb{N}\coloneqq\cb{1,2,\ldots}$. We assume that the standard Euclidean space $\R^n$ is equipped with an arbitrary norm $\abs{\cdot}$ and the usual inner product defined by $x^{\mathsf{T}}z\coloneqq \sum_{i=1}^n x_iz_i$ for $x,z\in\R^n$. We denote by $\R^n_+$ the positive orthant in $\R^n$, that is, the set of all $x=(x_1,\ldots,x_n)^{\mathsf{T}}\in\R^n$ with $x_i\geq 0$ for each $i\in\cb{1,\ldots,n}$.

To introduce the probabilistic setup, let $(\O,\F,\Pr)$ be a probability space and denote by $L_n^0$ the set of all $\F$-measurable random variables taking values in $\R^n$, where two elements are considered identical if they are equal $\Pr$-almost surely. For each $X\in L_n^0$, we define
\[
\norm{X}_p\coloneqq (\E\sqb{\abs{X}^p})^\frac1p
\]
for $p\in [1,+\infty)$ and
\[
\norm{X}_p \coloneqq \inf\{c\geq 0\mid \Pr\{\abs{X}\leq c\}=1\}
\]
for $p=+\infty$. Let $p\in [1,+\infty]$. The space $L_n^p\coloneqq \{X\in L_n^0\mid \norm{X}_p<+\infty\}$ is a Banach space equipped with the norm $\norm{\cdot}_p$. For brevity, let $L^p\coloneqq L_1^p$; for $Y_1,Y_2\in L^p$, we write $Y_1\leq Y_2$ if $\Pr\cb{Y_1\leq Y_2}=1$, which yields the cone $L^p_+\coloneqq\cb{Y\in L^p\mid 0\leq Y}$.

Let $\Y=L^p$, where $p\in [1,+\infty]$. The space $\Y$ is considered with its strong topology induced by the norm $\norm{\cdot}_p$ if $p<+\infty$ and with the weak$^\ast$ topology $\sigma(L^\infty,L^1)$ if $p=+\infty$. Under this topology, we denote by $\Y^\ast$ the topological dual space of $\Y$ with the bilinear duality mapping $\ip{\cdot,\cdot}\colon\Y^\ast\times\Y\to\R$. Hence, $\Y^\ast = L^q$ with $\ip{V,Y}=\E\sqb{VY}$ for every $V\in \Y^\ast, Y\in \Y$, where $q\in [1,+\infty]$ is the conjugate exponent of $p$, that is, $\frac1p+\frac1q=1$. In all cases, we consider $\Y^\ast$ with the weak topology $\sigma(\Y^\ast,\Y)$.

For a functional $\rho\colon\Y\to \bar{\R}\coloneqq [-\infty,+\infty]$, let us consider the following properties.
\begin{enumerate}[\bf (i)]
	\item \textbf{Monotonicity:} $Y_1 \leq Y_2$ implies $ \rho(Y_1) \geq \rho(Y_2)$ for every $Y_1,Y_2\in \Y$.
	\item \textbf{Quasiconvexity:} It holds $\rho(\lambda Y_1+(1-\lambda)Y_2)\leq\max\cb{\rho(Y_1),\rho(Y_2)}$ for every $Y_1,Y_2\in \Y$ and $\lambda\in(0,1)$.
	\item \textbf{Convexity:} It holds $\rho(\lambda Y_1+(1-\lambda)Y_2)\leq \lambda \rho(Y_1)+(1-\lambda)\rho(Y_2)$ for every $Y_1,Y_2\in\Y$ and $\lambda\in(0,1)$ (with the inf-addition convention $(+\infty)+(-\infty)=+\infty$ for the right hand side).
	\item \textbf{Translativity:} It holds $\rho(Y + y)=\rho(Y)-y$ for every $Y\in \Y$ and $y\in\R$.
	\item \textbf{Positive homogeneity:} It holds $\rho(\lambda Y)=\lambda \rho(Y)$ for every $Y\in \Y$ and $\lambda \geq 0$.
\end{enumerate}
For each $t\in\R$, let us define
\[
\A^t\coloneqq \cb{Y\in\Y\mid \rho(Y)\leq t},
\]
which is called the \emph{acceptance set} of $\rho$ at level $t$. It is easy to check that
\begin{equation}\label{recovery}
	\rho(Y)=\inf\cb{t\in\R\mid Y\in\A^t},\quad Y\in\Y.
\end{equation}
Note that $\rho$ satisfies quasiconvexity if and only if $\A^t$ is convex for each $t\in\R$. The functional $\rho$ is called a \emph{quasiconvex risk measure} if it satisfies (i) and (ii), a \emph{convex risk measure} if it satisfies (i), (ii), (iii), (iv), and a \emph{coherent risk measure} if it satisfies (i), (ii), (iii), (iv), (v). Hence, a coherent risk measure is necessarily a convex risk measure, and a convex risk measure is necessarily a quasiconvex risk measure. In addition, convexity implies quasiconvexity; and, under monotonicity, quasiconvexity and translativity imply convexity; see \citet[Exercise 4.1.1]{fs:sf}. Hence, when working within the general framework of quasiconvex risk measures, one does not assume translativity.

In the current paper, we study a static portfolio optimization problem with two quasiconvex risk measures. In order to formulate a dual problem associated to the portfolio optimization problem, the dual representations of these risk measures will have a crucial role. We review the dual representation result of \cite{drapeaukupper} next. To that end, let $\Y^\ast_+\coloneqq L^q_+$.

\begin{definition}\label{maximalpen}
	\cite[Definition~9]{drapeaukupper} A function $\beta\colon (\Y^\ast_+\sm\cb{0})\times\R\to \bar{\R}$ is called a \emph{maximal risk function} if it satisfies the following properties.
	\begin{enumerate}[(a)]
		\item $\beta$ is increasing and left-continuous in the second argument.
		\item $\beta$ is jointly quasiconcave.
		\item It holds $\beta(\lambda V,\lambda s)=\beta(V,s)$ for every $V\in\Y^\ast_+, s\in\R,\lambda>0$.
		\item It holds $\lim_{s\rightarrow-\infty}\beta(V_1,s)=\lim_{s\rightarrow-\infty}\beta(V_2,s)$ for every $V_1,V_2\in\Y^\ast_+$.
		\item The right-continuous version $(V,s)\mapsto \beta^+(V,s)\coloneqq \inf_{s^\prime>s}\beta(V,s^\prime)$ of $\beta$ is (weakly) upper semicontinuous in the first argument.
	\end{enumerate}
\end{definition}

Thanks to property (c) in Definition~\ref{maximalpen}, one can simply work with the restriction of $\b$ on the set $\Y_+^{\ast,1}\times\R$, where
\begin{equation}\label{basecone}
	\Y_+^{\ast,1}\coloneqq L^{q,1}_+\coloneqq \cb{Y\in L^q_+\mid \E[Y]=1}.
\end{equation}

Let $\rho$ be a lower semicontinuous quasiconvex risk measure. We define the \emph{minimal penalty function} $\alpha\colon \Y_+^\ast\times\R\to\bar{\R}$ of $\rho$ by
\[
\a(V,t)\coloneqq \sup_{Y\in\A^t}\E\sqb{-VY},\quad V\in\Y_+^\ast,t\in\R.
\]
For each $t\in\R$, the acceptance set $\A^t$ is a closed convex subset of $\Y$, hence it is characterized by its support function $V\mapsto \a(V,t)$. Consequently, the risk measure $\rho$ is uniquely determined by its minimal penalty function $\a$. In quasiconvex analysis, it is sometimes useful to work with the left-continuous version $\a^{-}$ of $\a$ (with respect to the second variable) defined by
\begin{equation}\label{support}
	\a^{-}(V,t)\coloneqq \sup_{t^\prime<t}\a(V,t^\prime),\quad V\in\Y_+^\ast, t\in\R.
\end{equation}
From the proof of \citet[Theorem~3]{drapeaukupper}, it follows that
\[
\a^{-}(V,t)=\sup_{Y\in\A^{t-}}\E\sqb{-VY},\quad V\in\Y_+^\ast,t\in\R,
\]
where, for each $t\in\R$, $\A^{t-}$ is the \emph{strict acceptance set} at level $t$ defined by
\[
\A^{t-}\coloneqq\cb{Y\in\Y\mid \rho(Y)<t}.
\]
For convenience, we recall the following result from \cite{drapeaukupper}, which provides a precise formula to recover $\rho$ from $\a$ or $\a^{-}$.

\begin{proposition}\label{dualrep}
	\cite[Theorem~3]{drapeaukupper} Let $\rho \colon \Y\to \bar{\R}$ be a lower semicontinuous quasiconvex risk measure with minimal penalty function $\a$. Then, there exists a unique maximal risk function $\beta$ such that
	\[
	\rho(Y)=\sup_{V\in \Y^\ast_+\backslash\cb{0}}\beta(V,\E[-VY]),\quad Y\in \Y.
	\]
	Moreover, such $\beta$ is given as the left-continuous pseudo-inverse of $\alpha$ or that of $\alpha^{-}$, that is,
	\[
	\beta(V,s)=\inf\cb{t\in\R\mid s\leq \a(V,t)}=\inf\cb{t\in\R\mid s\leq \a^{-}(V,t)}
	\]
	for every $V\in\Y^\ast_+$ and $s\in\R$.
\end{proposition}

\section{Portfolio optimization problem}\label{arbitrary}

In this section, we formulate and study a risk-averse portfolio optimization problem with two quasiconvex risk measures.

We suppose that there are $n$ risky assets with possibly correlated returns in a static model as in \cite{deniz}. To that end, let $X=(X_1,\ldots,X_n)^{\mathsf{T}}\in L_n^p$ be a random vector, where $p\in [1,+\infty]$. For each $i\in\cb{1,\ldots,n}$, the component $X_i$ denotes the return of the $i^{\text{th}}$ asset as a multiple of the initial price of that asset. In this static model, a portfolio is naturally defined as a vector $w\in\R^n$ with $\sum_{i=1}^n w_i=1$, where, for each $i\in\cb{1,\ldots,n}$, $w_i$ denotes the weight of the corresponding asset in the portfolio based on the asset prices at the beginning of the investment period. Hence, prohibiting shortselling, the set of all portfolios is defined as
\begin{equation}\label{defnw}
	\W\coloneqq\cb{w\in\R^n_+\mid \1^{\mathsf{T}}w = 1},
\end{equation}
where $\1=(1,\ldots,1)^{\mathsf{T}}\in\R^n$.
%When shortselling is not allowed, we will restrict ourselves to the portfolios in the subset
%\begin{equation}\label{posw}
%\W_+\coloneqq\W\cap \R^n_+,
%\end{equation}
%which is the $(n-1)$-dimensional unit simplex.
Given a portfolio $w\in\W$, note that the corresponding return $w^{\mathsf{T}}X$ is in the space $\Y=L^p$.

To model risk-aversion, suppose that we have two proper lower semicontinuous quasiconvex risk measures $\rho_1,\rho_2\colon L^p\to\bar{\R}$. The portfolio manager aims to choose a portfolio $w\in\W$ that minimizes the type 1 risk $\rho_1(w^{\mathsf{T}}X)$ while satisfying the type 2 risk constraint
\[
\rho_2(w^{\mathsf{T}}X)\leq r,
\]
where $r\in\R$ is a fixed threshold for this type of risk. Hence, the porfolio optimization problem of interest is formulated as
\begin{align*}
	&\text{minimize}\; \;\;                  \rho_1(w^{\mathsf{T}}X)\tag{$\P(r)$}\\
	&\text{subject to}\;\;                   \rho_2(w^{\mathsf{T}}X)\leq r \\
	& \quad\quad\quad\quad\;\;\;      w\in \W.
\end{align*}
We denote by $p(r)$ the optimal value of $(\P(r))$.

As in \cite{deniz}, the motivation to use two risk measures comes from decision making under two risk perceptions. For instance, the portfolio manager may choose the portfolio by using $\rho_1$ as the suitable risk measure for her risk perception but there might be an obligation to consider the opinion of a regulatory authority whose risk perception is encoded by $\rho_2$. In another setting, the portfolio manager may wish to work with two risk measures but the principle risk measure $\rho_1$ may have higher seniority than $\rho_2$. In particular, this framework covers as special cases the problem of maximizing expected return subject to a quasiconvex risk constraint (arbitrary $\rho_2$) if we take $\rho_1(Y)=\E\sqb{-Y}$ for each $Y\in L^p$, as well as the problem of minimizing a quasiconvex risk measure (arbitrary $\rho_1$) while maintaining a sufficiently high expected return if we take $\rho_2(Y)=\E\sqb{-Y}$ for each $Y\in L^p$. In the general case, while the previous work \cite{deniz} is restricted to coherent risk measures (e.g., negative expected value, average value-at-risk), the current work covers a much larger class of risk measures as it allows for both convex (translative) risk measures (e.g., entropic risk measure, utility-based shortfall risk measures) as well as quasiconvex and non-convex (non-translative) risk measures (e.g., certainty equivalent, economic index of riskiness). Some detailed examples will be studied in Section \ref{examples}.

For each $j\in\cb{1,2}$, let us denote by $\A_j^t$ 
%($\A^{t-}$) the (strict) 
the acceptance set of $\rho_j$ at level $t\in\R$ and by $\a_j$ 
%($\a_j^{-}$) 
the 
%(left-continuous version of the) 
minimal penalty function of $\rho_j$. In view of Proposition~\ref{dualrep}, there exists a unique maximal risk function $\b_j$ such that $\rho_j$ admits the dual representation
\begin{equation}\label{dualrep2}
	\rho_j(Y)=\sup_{V\in L^q_+\backslash\cb{0}}\b_j(V,\E\sqb{-VY}),\quad Y\in L^p,
\end{equation}
and we have $\beta_j(V,s)=\inf\cb{t\in\R\mid s\leq \a_j(V,t)}$ %=\inf\{t\in\R\mid s\leq \a^{-}_j(V,t)\}$ 
for every $V\in L^q_+$, $s\in\R$. Let us also define a function $g_{j}\colon\R^n\to\bar{\R}$ by
\begin{equation}\label{g1}
	g_j(w)= \rho_j(w^{\mathsf{T}}X), \quad w\in\R^n,
\end{equation}
which is a quasiconvex function.

We present the analysis of $(\P(r))$ in the next two sections.% The first one (Theorem \ref{thm1} in Section \ref{sect1}) associates to $(\P(r))$ a new problem that can be viewed as a dual problem, and shows that the optimal values of these two problems coincide. For the second main theorem (Theorem \ref{thm2} in Section \ref{sect2}), we assume that $\a_1$ satisfies an additional structural property (see Assumption \ref{tildealphaconcave}) under which the dual problem becomes a convex optimization problem, and show that an optimal Lagrange multiplier attached to the dual problem is an optimal solution of $(\P(r))$.

\section{Analysis of the problem when the principle risk measure is convex}\label{convexcase}

In this section, we analyze the porfolio optimization problem $(\P(r))$ under the following assumption.

\begin{assumption}\label{rho1convex}
	$\rho_1$ satisfies convexity.
\end{assumption}

Note that Assumption~\ref{rho1convex} does \emph{not} ensure that $\rho_1$ is a convex risk measure in general as the latter terminology assumes translativity in addition to convexity (see Section \ref{setup}).

\subsection{First main result: establishing strong duality}\label{sect1}

As a preparation for Theorem~\ref{thm1}, we first introduce an auxiliary problem that will show up in the derivation of the dual problem. For each $V_1,V_2\in L^q_+$, let us define
\begin{align}
	h^P(V_1,V_2)&\coloneqq \inf\cb{t+\E[-V_2X]^{\mathsf{T}}w\mid \a^{-}_1(V_1,t)\geq \E[-V_1X]^{\mathsf{T}}w,\ w\in\W,\ t\in\R}\label{auxh}.
\end{align}

For the first result of this section (Proposition~\ref{prop1} below), we need the following continuity assumption on the maximal risk function of $\rho_1$; note that it also appears as part of \citet[Assumption~3.2]{sigrid} in the context of a dynamic portfolio optimization problem. Recall \eqref{basecone} for the definition of $L^{q,1}_+$.

\begin{assumption}\label{additionalbeta}
	$\b_1$ is jointly (weakly) upper semicontinuous on $L^{q,1}_+\times\R$.
\end{assumption}

Before stating Proposition~\ref{prop1}, we introduce a restrictive finiteness assumption which will be removed later in Section \ref{sect2}; see Remark~\ref{finrem} for a discussion on this assumption.

\begin{assumption}\label{finiteness}
	For every $w\in\W$, $V_1\in L^{q,1}_+$ and $V_2\in L^q_+$, it holds $\b_1(V_1,\E[-V_1w^{\mathsf{T}}X])\in\R$ and $\a_2(V_2,r)\in\R$.
\end{assumption}

The next result establishes the connection between $(\P(r))$ and $h^P$.

\begin{proposition}\label{prop1}
	Suppose that Assumption~\ref{rho1convex}, Assumption~\ref{additionalbeta}, Assumption~\ref{finiteness} hold. Then,
	\[
	p(r)=\sup_{V_1,V_2\in L^q_+}\of{h^P(V_1,V_2) -\a_2(V_2,r)}.
	\]
\end{proposition}

\begin{proof}
	Let $I_{\A_2^r}$ be the convex analytic indicator function of $\A_2^r$, that is, $I_{\A_2^r}(Y)=0$ whenever $Y\in \A_2^r$ and $I_{\A_2^r}(Y)=+\infty$ whenever $Y\in L^p\setminus\A_2^r$. Then, using \citet[Theorem 2.3.3, (2.33)]{zalinescu}, Proposition \ref{dualrep} and property (c) in Definition \ref{maximalpen}, we obtain
	\begin{align*}
		p(r)&=\inf\cb{\rho_1(w^\mathsf{T}X)\mid \rho_2(w^{\mathsf{T}}X)\leq r,\ w\in\W}\\
		&=\inf\cb{\rho_1(w^\mathsf{T}X)\mid w^{\mathsf{T}}X\in\A_2^r,\ w\in\W}\\
		&=\inf_{w\in\W}\of{\rho_1(w^{\mathsf{T}}X)+I_{\A_2^r}(w^{\mathsf{T}}X)}\\
		&=\inf_{w\in\W}\of{\sup_{V_1\in L_+^{q,1}}\b_1\of{V_1,\E[-V_1w^{\mathsf{T}}X]}+\sup_{V_2\in L_+^q}\of{ \E[-V_2w^\mathsf{T}X]-\a_2(V_2,r)}}\\
		&=\inf_{w\in\W}\sup_{V_1\in L^{q,1}_+,V_2\in L^{q}_+}b(w,V_1,V_2),
	\end{align*}
	where
	\[
	b(w,V_1,V_2)\coloneqq \b_1\of{V_1,\E[-V_1w^{\mathsf{T}}X]}+ \E[-V_2w^\mathsf{T}X]-\a_2(V_2,r).
	\]
	For fixed $w\in\W, V_1\in L^{q,1}_+$, clearly $V_2\mapsto \E[-V_2w^\mathsf{T}X]-\a_2(V_2,r)$ is concave and upper semicontinuous by the properties of support function. Let $w\in\W, V_2\in L^{q}_+$. For each $V_1,V_1^\prime\in L^{q,1}_+$ and $\lambda\in (0,1)$, we have
	\begin{align*}
		&\b_1\of{\lambda V_1+(1-\lambda)V_1^\prime,\E[-(\lambda V_1+(1-\lambda)V_1^\prime)w^{\mathsf{T}}X]}\\
		&=\b_1\of{\lambda V_1+(1-\lambda)V_1^\prime,\lambda\E[- V_1w^{\mathsf{T}}X]+(1-\lambda)\E[- V^\prime_1w^{\mathsf{T}}X]}\\
		&\geq \min\cb{\b_1\of{V_1,\E[-V_1w^{\mathsf{T}}X]},\b_1\of{V^\prime_1,\E[-V^\prime_1w^{\mathsf{T}}X]}}
	\end{align*}
	by the joint quasiconcavity of $\b_1$. Hence, the function $V_1\mapsto \b_1\of{V_1,\E[-V_1w^{\mathsf{T}}X]}$ is quasiconcave. We claim that this function is also weakly upper semicontinuous. To that end, let $(V_\a)_{\a\in I}$ be a weakly convergent net in $L_+^{q,1}$ with some index set $I$ and limit $V$. Hence, $\lim_{\a\to\infty}\E[-V_\a w^{\mathsf{T}}X]=\E[-V w^{\mathsf{T}}X]$. So $(V_\a, \E[-V_\a w^{\mathsf{T}}X])_{\a\in I}$ is a weakly convergent net in $L^{q,1}_+\times\R$ with limit $(V,\E[-V w^{\mathsf{T}}X])$. By Assumption~\ref{additionalbeta}, we get
	\[
	\limsup_{\a\to\infty}\b_1(V_\a,\E[-V_\a w^{\mathsf{T}}X])\leq \b_1(V,\E[-V w^{\mathsf{T}}X]).
	\]
	Hence, the claim follows. Since $\beta_1$ is increasing and left-continuous in the second argument, it is also lower semicontinuous (indeed continuous by Assumption~\ref{additionalbeta}). Hence, for fixed $V_1\in L^{q,1}_1, V_2\in L^q_+$, the function $w\mapsto \b_1\of{V_1,\E[-V_1w^{\mathsf{T}}X]}$ is lower semicontinuous (indeed continuous) and quasiconvex. Note that $\W$ is a convex and compact set. Finally, by Assumption~\ref{finiteness}, $b$ is real-valued on $\W\times L^{q,1}_+\times L^{q}_+$. Therefore, by Sion's minimax theorem \cite[Corollary~3.3]{sion}, we may write
	\begin{align*}
		p(r)&=\sup_{V_1\in L^{q,1}_+, V_2\in L^q_+}\inf_{w\in\W}b(w,V_1,V_2)\\
		&=\sup_{V_1\in L^{q,1}_+, V_2\in L^q_+}\of{\inf_{w\in\W}\of{\b_1(V_1,\E[-V_1X]^{\mathsf{T}}w) +\E[-V_2X]^{\mathsf{T}}w}-\a_2(V_2,r)}\\
		&=\sup_{V_1, V_2\in L^q_+}\of{\inf_{w\in\W}\of{\b_1(V_1,\E[-V_1X]^{\mathsf{T}}w) +\E[-V_2X]^{\mathsf{T}}w}-\a_2(V_2,r)}\\
		&=\sup_{V_1\in V_2\in L^q_+}\of{\inf_{w\in\W}\of{\inf\cb{t\in\R\mid \a^-_1(V_1,t)\geq \E[-V_1X]^{\mathsf{T}}w} +\E[-V_2X]^{\mathsf{T}}w}-\a_2(V_2,r)}\\
		&=\sup_{V_1,V_2\in L^{q}_+}\of{h^P(V_1,V_2) -\a_2(V_2,r)},
	\end{align*}
	where $h^P(V_1,V_2)$ is defined by \eqref{auxh}.
\end{proof}

\begin{remark}\label{finrem}
	In the proof of Proposition~\ref{prop1}, the only role of Assumption~\ref{finiteness} is to ensure that the function $b$ has finite values so that Sion's minimax theorem can be applied. In cases where Assumption~\ref{finiteness} is not valid, the proof of Proposition~\ref{prop1} can be seen as a heuristic argument to come up with a dual formulation of $(p(r))$. Thanks to Theorem~\ref{thm2} below, it turns out that the conclusion of Proposition~\ref{prop1} is still valid without Assumption~\ref{finiteness}.
\end{remark}

Let $V_1,V_2\in L^q_+$. Note that $h^P(V_1,V_2)$ is the optimal value of a finite-dimensional optimization problem which is in general nonconvex due to the inequality constraint when $\rho_1$ is only assumed to be a quasiconvex risk measure. However, under Assumption~\ref{rho1convex}, as argued in the proof of \citet[Proposition~6]{drapeaukupper}, the function $t\mapsto \a_1^-(V_1,t)$ is concave and the optimization problem in $h^P(V_1,V_2)$ becomes convex. Before proceeding further, we introduce an additional assumption related to the asymptotic behavior of $\a_1^-$. %However, as it will be clear in Section \ref{examples}, this assumption is violated by many famous examples of quasiconvex risk measures that are not convex. Instead, most of these examples and also all convex (translative) risk measures satisfy the following assumption.

\begin{assumption}\label{assmpalpha}
	For each $V_1\in L^q_+$, it holds
	\[
	\lim_{t\rightarrow\infty}\a^-_1(V_1,t)=+\infty.
	\]
\end{assumption}

To formulate the next proposition, let us define
\[
\a^\ast_1(V_1,z)\coloneqq\inf_{t\in\R}\of{tz-\a^-_1(V_1,t)}=\inf_{t\in\R}\of{tz-\a_1(V_1,t)},\quad V_1\in L^q_+, z\in\R.
\]
Note that $z\mapsto -\a_1^\ast(V_1,-z)$ is the conjugate function of $\a_1^-$ (and also of $\a_1$) with respect to the second variable. Since $t\mapsto \a_1^{-}(V_1,t)$ is an increasing function, it is easy to check that $\a_1^{\ast}(V_1,z)=-\infty$ whenever $z<0$. Moreover, under Assumption~\ref{rho1convex}, for each $V_1\in L^q_+$, the function $t\mapsto \a_1^-(V_1, t)$ is concave and upper semicontinuous so that Fenchel-Moreau theorem \cite[Theorem~2.3.3]{zalinescu} gives
\begin{equation}\label{fm}
	\a_1^-(V_1,t)=\inf_{z\geq 0}\of{tz-\a_1^\ast(V_1,z)},\quad V_1\in L^q_+, t\in\R.
\end{equation}
The following special value of $\a_1^\ast$ will play an important role in the dual problem of $(\P(r))$:
\begin{equation}\label{tildealpha}
	\tilde{\a}_1(V_1)\coloneqq \a_1^\ast(V_1,1)=\inf_{t\in\R}\of{t-\alpha^-_1(V_1,t)},\quad V_1\in L^q_+.
\end{equation}
Using this quantity, let us define, for each $V_1,V_2\in L^q_+$,
\[
h^D(V_1,V_2)\coloneqq \sup\cb{\tilde{\a}_1(xV_1)-y\mid x\E[V_1X]+\E[V_2X]\leq y\1,\ x> 0,\ y\in\R},
\]
where the inequality constraint is understood in the componentwise manner. As the proof of the next result shows, the dual of the problem in $h^P(V_1,V_2)$ gives rise to $h^D(V_1,V_2)$.

%\begin{assumption}\label{slater}
%	For every $V_1\in L^q_+$, there exists $w\in\W$ and $t\in\R$ such that $\a_1^-(V,t)>\E[-V_1X]^{\mathsf{T}}w$.
%	\end{assumption}

\begin{proposition}\label{prop2}
	Suppose that Assumption~\ref{rho1convex} and Assumption~\ref{assmpalpha} hold, and let $V_1,V_2\in L^q_+$. Then,
	\[
	h^P(V_1,V_2)=h^D(V_1,V_2).
	\]
\end{proposition}

\begin{proof}
	As noted above, the problem defining $h^P(V_1,V_2)$ is a convex optimization problem under Assumption~\ref{rho1convex}. By Assumption~\ref{assmpalpha}, there exists $w^0\in\W$ and $t^0\in\R$ such that $\a_1^-(V_1,t^0)>\E[-V_1X]^{\mathsf{T}}w^0$ (indeed, for every $w^0\in\W$ such $t^0\in\R$ exists). Hence, Slater's condition holds for this convex optimization problem and we have
	\begin{align*}
		h^P(V_1,V_2)&=\sup_{\substack{x\geq 0,\\ y\in\R}}\inf_{\substack{w\in\R^n_+,t\in\R:\\ 
				\a^-_1(V_1,t)\in\R}}\of{t+\E[-V_2X]^{\mathsf{T}}w+x(-\a_1^-(V_1,t)+\E[-V_1X]^{\mathsf{T}}w)+y(\1^{\mathsf{T}}w-1)}\\
		&=\sup_{\substack{x\geq 0,\\ y\in\R}}\inf_{\substack{w\in\R^n_+,t\in\R:\\ 
				\a^-_1(V_1,t)\in\R}}\of{t-x\a_1^-(V_1,t)+(x\E[-V_1X]+\E[-V_2X]+y\1)^{\mathsf{T}}w-y}\\
		&=\sup_{\substack{x> 0,\\ y\in\R}}\inf_{\substack{w\in\R^n_+,t\in\R:\\ 
				\a^-_1(V_1,t)\in\R}}\of{t-\a_1^-(xV_1,t)+(x\E[-V_1X]+\E[-V_2X]+y\1)^{\mathsf{T}}w-y}\\
		&=\sup_{\substack{x> 0,\\ y\in\R}}\cb{\tilde{\a}_1(xV_1)-y\mid x\E[V_1X]+\E[V_2X]\leq y\1}.
	\end{align*}
	Let us justify each passage in the above derivation: the first equality is by strong duality for convex optimization, the second is by simple manipulations, the third excludes the case $x=0$ from further consideration since in this case the infimum yields $-\infty$, the fourth is by evaluating the infimum with respect to $w\in\R^n_+$ and $t\in\R$. Therefore, $h^P(V_1,V_2)=h^D(V_1,V_2)$.
\end{proof}

We are ready to prove the first main result of the paper, which establishes strong duality between $(\P(r))$ and a new problem, which we refer to as the \emph{dual problem} of $(\P(r))$.

\begin{theorem}\label{thm1}
	Suppose that Assumption~\ref{rho1convex}, Assumption~\ref{additionalbeta}, Assumption~\ref{finiteness}, Assumption~\ref{assmpalpha} hold, and consider the problem
	\begin{align*}
		%		&\text{maximize}\; \;               t-\a_1(V_1,t)-\a_2(V_2,r)-y\tag{$\D(r)$}\\
		&\text{maximize}\; \;               \tilde{\a}_1(V_1)-\a_2(V_2,r)-y\tag{$\D(r)$}\\
		&\text{subject to}\;\;                   \E[V_1X]+\E[V_2X]\leq y\1\\
		%	& \quad\quad\quad\quad\;\;\;  \tilde{\a}_1(V_1)\in\R \\
		& \quad\quad\quad\quad\;\;\;      V_1, V_2\in L^q_+,\; y\in\R,
	\end{align*}
	where the inequality constraint is understood in the componentwise sense. Then, $(\P(r))$ and $(\D(r))$ have the same optimal value $p(r)$.
\end{theorem}

\begin{proof}
	Combining Proposition~\ref{prop1} and Proposition~\ref{prop2}, we obtain
	\begin{align*}
		p(r) 
		&=\sup_{V_1,V_2\in L^q_+}\of{h^D(V_1,V_2) -\a_2(V_2,r)}\\
		&=\sup_{V_1,V_2\in L^q_+}\of{\sup\cb{\tilde{\a}_1(xV_1)-y\mid \E[xV_1X]+\E[V_2X]\leq y\1,\ x>0,\ y\in\R} -\a_2(V_2,r)}\\
		&=\sup\cb{\tilde{\a}_1(xV_1)-\a_2(V_2,r)-y\mid  \E[xV_1X]+\E[V_2X]\leq y\1,  V_1,V_2\in L^q_+, x> 0, y\in\R}\\
		&=\sup\cb{\tilde{\a}_1(V_1)-\a_2(V_2,r)-y\mid  \E[V_1X]+\E[V_2X]\leq y\1,\  V_1,V_2\in L^q_+,\ y\in\R},
		%  &=\sup\cb{t-\a_1(V_1,t)-\a_2(V_2,r)-y\mid  \E[V_1X]+\E[V_2X]\leq y\1,\ V_1,V_2\in L^q_+,\ t,y\in\R},
	\end{align*}
	which coincides with the optimal value of $(\D(r))$.
\end{proof}

It is worth noting that the dual problem $(\D(r))$ is a convex optimization problem thanks to Assumption~\ref{rho1convex}.
%in general since $\tilde{\a}_1$ may fail to be a concave function and $\{V_1\in L^q_+\mid \tilde{\a}_1(V_1)\in\R\}$ may fail to be a convex set. This is as expected due to the non-convex nature of the primal problem $(\P(r))$. Therefore, solving $(\D(r))$ in practice might be a challenging task without further structure. In the next subsection, we exploit a case where $(\D(r))$ is a convex optimization problem by imposing additional assumptions on $\tilde{\a}_1$.

\subsection{Second main result: establishing optimality}\label{sect2}

While Theorem~\ref{thm1} provides strong duality between $(\P(r))$ and $(\D(r))$, it does not make a statement on how to find an optimal portfolio $w^\ast\in\W$ for $(\P(r))$. The aim of Theorem~\ref{thm2}, the second main result of the paper, is to find such $w^\ast$ in relation to $(\D(r))$.

As a preparation, we recall some well-known concepts and facts from convex analysis.
% as well as some lesser-known tools of quasiconvex analysis. 
To that end, let us fix an arbitrary Hausdorff locally convex topological vector space $\X$ with topological dual $\X^\ast$ and bilinear duality mapping $\ip{\cdot,\cdot}\colon\X^\ast\times\X\to\R$. For our purposes, the following special cases of $\X$ are particularly important:
\begin{enumerate}[(i)]
	\item $\X=\R^n$ with the usual topology, which yields $\X^\ast=\R^n$ together with $\ip{z,x}=z^{\mathsf{T}}x$ for every $x, z\in\R^n$.
	\item $\X=L^q$ with $q\in [1,+\infty)$ with the weak topology $\sigma(L^q,L^p)$, which yields $\X^\ast=L^p$ together with $\ip{Y,U}=\E\sqb{UY}$ for every $U\in L^q$, $Y\in L^p$.
	\item $\X=L^\infty$ with the weak topology $\sigma(L^\infty,L^1)$, which yields $\X^\ast=L^1$ together with $\ip{Y,U}=\E\sqb{UY}$ for every $U\in L^\infty$, $Y\in L^1$.
\end{enumerate}

Consider a set $A\subseteq \X$. The function $I_A\colon\X\to\R\cup\cb{+\infty}$ defined by
\[
I_A(x)=\begin{cases}0 & \text{if }x\in A,\\
	+\infty &\text{if }x\in \X\sm A,
\end{cases}
\]
is called the \emph{indicator function} of $A$; note that $A$ is a convex set if and only if $I_A$ is a convex function. For a point $x\in A$, %$\cone(A)\coloneqq \cb{\lambda x\mid \lambda\geq 0, x\in A}$ is called the \emph{conic hull} of $A$. If $A$ is convex, then $\cone(A)$ is a convex cone. 
%For $x\in A$, 
the convex cone
\[
\mathcal{N}(A,x)\coloneqq\cb{z\in\X^\ast\mid \forall x^\prime\in A\colon \ip{z,x}\geq \ip{z,x^\prime}}
\]
is called the \emph{normal cone} of $A$ at $x$.  Let $g\colon \X\to\R\cup\cb{+\infty}$ be a function. Given $x\in \X$, the set 
\[
\partial g(x)\coloneqq\cb{z\in\X^\ast\mid \forall x^\prime\in \X\colon g(x^\prime)\geq g(x)+\ip{z,x^\prime-x}}
\]
is called the \emph{subdifferential} of $g$ at $x$. If $A$ is a nonempty convex set, then by \citet[Section~2.4]{zalinescu}, $\partial I_A(x)=\N_A(x)$ for every $x\in A$, and $\partial I_A(x)=\emptyset$ for every $x\in\X\sm A$. The function $g^\ast\colon\X^\ast\to\bar{\R}$ defined by $g^\ast(z)\coloneqq\sup_{x\in\X}\of{\ip{z,x}-g(x)}$ for each $z\in\X^\ast$ is called the \emph{conjugate function} of $g$. We have
\begin{equation}\label{conjsub}
	z\in\partial g(x) \quad \Leftrightarrow\quad x\in\partial g^{\ast}(z)
\end{equation}
for every $x\in\X,z\in\X^\ast$ such that $g$ is lower semicontinuous at $x$. If $A$ is a nonempty closed convex set, then it is well-known that $\sigma_A\coloneqq (I_A)^\ast$ is the support function of $A$ defined by
\[
\sigma_A(z)\coloneqq \sup_{x\in A}\ip{z,x},\quad z\in\X^\ast.
\]
From the above definitions, it is clear that, for a point $x\in A$, we have
\begin{equation}\label{normalch}
	\N(A,x)=\cb{z\in\X^\ast\mid \sigma_A(z)=\ip{z,x}}.
\end{equation}

Consider the problem of minimizing $g$ over $A$. Suppose that $g$ is convex and let $x\in A$ with $g(x)<+\infty$. By Pshenichnyi-Rockafellar theorem \cite[Theorem~2.9.1]{zalinescu}, if
\begin{equation}\label{PRthm}
	\partial g(x)\cap -\N(A,x)\neq\emptyset,
\end{equation}
then $x$ is a minimizer of $g$ over $A$, that is, $g(x)=\inf_{x^\prime\in A}g(x^\prime)$; the converse also holds if $g$ is continuous at $x$.

The next lemma is devoted to the calculation of a certain subdifferential that is relevant to $(\P(r))$.

\begin{lemma}\label{GPcalc}
	Suppose that Assumption~\ref{rho1convex} holds. Let $w\in\R^n$ be such that $g_1(w)=\rho_1(w^{\mathsf{T}}X)<+\infty$. Then,
	\[
	\cb{\E\sqb{-VX}\mid \b_1(V,-w^{\mathsf{T}}\E[VX])=g_1(w),\ V\in L^q_+ }\subseteq\partial g_1(w).
	\]
\end{lemma}

\begin{proof}
	%Since $\rho_j$ is continuous from below, $\mathscr{V}_j(w)\neq \emptyset$.
	Thanks to Assumption~\ref{rho1convex}, the function $g_1$ defined by \eqref{g1} is convex. Let us define the continuous linear operator $L\colon\R^n\to L^p$ by
	\[
	Lw^\prime\coloneqq X^{\mathsf{T}}w^\prime,\quad w^\prime\in\R^n.
	\]
	Then, its adjoint operator $L^\ast\colon L^q\to \R^n$ is given by
	\[
	L^\ast V= \E\sqb{VX},\quad V\in L^q.
	\]
	Since $g_1=\rho_1\circ L$ and $\rho_1$ is finite at $w^{\mathsf{T}}X$, by \citet[Theorem~2.8.3(iii)]{zalinescu}, we have 
	\begin{equation}\label{GP1}
		\cb{L^\ast V\mid V\in\partial\rho_1(w^{\mathsf{T}}X)}=\cb{\E\sqb{VX}\mid V\in\partial\rho_1(w^{\mathsf{T}}X)}\subseteq\partial g_1(w).
	\end{equation}
	
	On the other hand, the subdifferential of the convex lower semicontinuous function $\rho_1$ at a point $Y\in L^p$ is given by
	\begin{equation}\label{GP2}
		\partial\rho_1(Y)=\cb{-V\mid \b_1(V,\E\sqb{-VY})=\rho_1(Y),\ V\in L^q_+},
	\end{equation}
	that is, it is the set of maximizers in the dual representation \eqref{dualrep2}. Taking $Y=w^{\mathsf{T}}X$ in \eqref{GP2} and combining it with \eqref{GP1} yields the claim of the lemma. 
\end{proof}

%Recalling \eqref{tildealpha}, we note that $\tilde{\a}_1$ may fail to be concave and the set $\{V_1\in L^q_+\mid \tilde{\a}_1=+\infty\}$ may be nonempty in general even under Assumption \ref{assmpalpha}. The next assumption ensures that $(\D(r))$ is a concave maximization problem by avoiding these complications.
%
%\begin{assumption}\label{tildealphaconcave}
%	The set $\{V_1\in L^q_+ \mid \tilde{\a}_1(V_1)\in \R\}$ is convex, and the function $\tilde{\a}_1$ is concave and upper semicontinuous on this set.
%\end{assumption}

The next assumption is a constraint qualification for $(\D(r))$.

\begin{assumption}\label{quasirelintnew}
	There exist $V_1,V_2\in L^q_{++}$ such that $\tilde{\a}_1(V_1)\in\R$, $\a_2(V_2,r)\in\R$.
\end{assumption}

We are ready to prove the second main theorem of the paper, which establishes the optimality of a Lagrange multiplier associated to $(\D(r))$.

\begin{theorem}\label{thm2}
	Suppose that Assumption~\ref{rho1convex}, Assumption~\ref{additionalbeta}, Assumption~\ref{assmpalpha}, Assumption~\ref{quasirelintnew} hold, and there exists an optimal solution $(V_1^\ast,V_2^\ast, y^\ast)\in L^q_+\times L^q_+\times \R$ for $(\D(r))$. Then, there exists an optimal Lagrange multiplier $w^\ast\in\R^n$ associated to the inequality constraint of $(\mathscr{D}(r))$. Moreover, every $w^\ast\in\R^n$ that is the Lagrange multiplier of the equality constraint of $(\mathscr{D}(r))$ at optimality is an optimal solution for $(\P(r))$, and $(\P(r))$ and $(\D(r))$ have the same optimal value $p(r)$.
\end{theorem}

\begin{proof}
	Let $(V_1^\ast,V_2^\ast, y^\ast)\in L^q_+\times L^q_+\times \R$ be an optimal solution for $(\D(r))$. Let us denote by $d(r)$ the optimal value of $(\D(r))$. By Assumption~\ref{quasirelintnew}, Slater's condition holds, that is, there exist $V_1,V_2\in L^q_{++}$, $y\in\R$ such that $\tilde{\a}_1(V_1)\in\R$, $\a_2(V_2,r)\in\R$ and
	\[
	\max_{i\in\cb{1,\ldots,n}}\of{\E\sqb{V_1X_i}+\E\sqb{V_2X_i}}< y
	\]
	since we may simply take $y:=\max_{i\in\cb{1,\ldots,n}}\of{\E\sqb{V_1X_i}+\E\sqb{V_2X_i}}+1$. Hence, by \citet[Corollary~4.8]{borwein}, there is zero duality gap between $(\D(r))$ and its Lagrange dual problem, and we may write
	\begin{align}
		d(r)&=\inf_{w\in\R_+^n}\sup_{V_1,V_2\in L^q_+,y\in\R}\of{\tilde{\a}_1(V_1)-\a_2(V_2,r)-y-w^{\mathsf{T}}\of{\E\sqb{V_1X}+\E\sqb{V_2X}-y\1}}\notag\\
		&=\inf_{w\in\R_+^n}\sup_{V_1,V_2\in L^q_+,y\in\R}\of{\tilde{\a}_1(V_1)-\a_2(V_2,r)-y+\E[-V_1w^{\mathsf{T}}X]+\E[-V_2w^{\mathsf{T}}X]+y w^{\mathsf{T}}\1}.\label{nested}
	\end{align}
	Moreover, \citet[Corollary~4.8]{borwein} also ensures that there exists an optimal Lagrange multiplier $w^\ast\in\R^n$ so that
	\[
	d(r)=\sup_{V_1,V_2\in L^q_+,y\in\R}\of{\tilde{\a}_1(V_1)-\a_2(V_2,r)-y+\E[-V_1(w^\ast)^{\mathsf{T}}X]+\E[-V_2(w^\ast)^{\mathsf{T}}X]+y (w^\ast)^{\mathsf{T}}\1},
	\]
	and $(V_1^\ast,V_2^\ast, y^\ast)$ is an optimal solution of the above concave maximization problem.
	
	Let $w\in\R_+^n$. Note that the inner (maximization) problem in \eqref{nested} is easily separated into three terms as
	\begin{equation}\label{three}
		\sup_{V_1\in L^q_+}\of{\tilde{\a}_1(V_1)+\E[-V_1w^{\mathsf{T}}X]}+\sup_{V_2\in L^q_+}\of{-\a_2(V_2,r)+\E[-V_2w^{\mathsf{T}}X]} + \sup_{y\in\R}y(1- w^{\mathsf{T}}\1).
	\end{equation}
	From the last term in \eqref{three}, it follows immediately that 
	\[
	\sup_{y\in\R}y(1- w^{\mathsf{T}}\1) = I_{\W}(w).
	\]
	In particular, if $w=w^\ast$, then we must have $w^\ast\in\W$. It is also easy to check that
	\begin{equation}\label{condy}
		y^\ast \1 \in \N(\W,w^\ast).
	\end{equation}
	
	For the first term in \eqref{three}, note that 
	\begin{align*}
		\rho_1(w^{\mathsf{T}}X)&=\sup_{V_1\in L^q_+}\b_1(V_1,\E[-V_1w^{\mathsf{T}}X])\\
		&=\sup_{V_1\in L^q_+}\inf\cb{t\in\R\mid \a^-_1(V_1,t)\geq \E[-V_1w^{\mathsf{T}}X]}\\
		&=\sup_{V_1\in L^q_+,x\geq 0}\inf_{\substack{t\in\R: \\ \a_1^-(V_1,t)\in\R}}\of{t+x\E[-V_1w^{\mathsf{T}}X]-x\a^-_1(V_1,t)}\\
		&=\sup_{V_1\in L^q_+,x\geq 0}\inf_{\substack{t\in\R: \\ \a_1^-(V_1,t)\in\R}}\of{t+\E[-xV_1w^{\mathsf{T}}X]-\a^-_1(xV_1,t)}\\
		&=\sup_{V_1\in L^q_+}\inf_{\substack{t\in\R: \\ \a_1^-(V_1,t)\in\R}}\of{t+\E[-V_1w^{\mathsf{T}}X]-\a^-_1(V_1,t)}\\
		&=\sup_{V_1\in L^q_+}\of{\tilde{\a}_1(V_1)+\E[-V_1w^{\mathsf{T}}X]}.
	\end{align*}
	The steps of this calculation are justified by following the same arguments as in the proof of Proposition \ref{prop2}, hence we omit this justification for brevity.
	%In this calculation, the first equality is by \eqref{dualrep2}, the second is by the definition of $\b_1$, the third is by strong duality for convex optimization since the following Slater's condition is satisfied thanks to Assumption \ref{assmpalpha}: there exists $t\in\R$ such that $\a_1(V_1,t)>\E[-V_1(w^\ast)^{\mathsf{T}}X]$. The rest of the calculation is straightforward.
	In particular, when $w=w^\ast$, by the optimality property of $V_1^\ast$ and property (c) of the definition of maximal risk function, we have
	\[
	g_1(w^\ast)=\rho_1((w^\ast)^{\mathsf{T}}X)=\E[-V^\ast_1(w^\ast)^{\mathsf{T}}X]+\tilde{\a}_1(V^\ast_1)=\b_1(V_1^\ast, \E[-V_1^\ast(w^\ast)^{\mathsf{T}}X]).
	\]
	By Lemma~\ref{GPcalc}, it follows that
	\begin{equation}\label{condy1}
		\E[-V_1^\ast X]\in \partial g_1(w^\ast).
	\end{equation}
	
	For the second term in \eqref{three}, we first note that $V_2\mapsto \a_2(V_2,r)$ is closely related to the support function of the closed convex set $\A_2^r$; indeed, we have
	\[
	\a_2(V_2,r)=\sup_{Y\in\A_2^r}\E[-V_2Y]=\sigma_{\A_2^r}(-V_2),\quad V_2\in L^q_+.
	\]
	Hence, by the conjugate duality between indicator function and support function, we have
	\[
	\sup_{V_2\in L^q_+}\of{-\a_2(V_2,r)+\E[-V_2w^{\mathsf{T}}X]} = I_{\A_2^r}(w^{\mathsf{T}}X)=I_{\{g_2\leq r\}}(w).
	\]
	In particular, when $w=w^\ast$, by the first-order condition, we have
	\begin{equation}\label{conda}
		-(w^\ast)^{\mathsf{T}}X\in\partial \a_2(V^\ast_2,r),
	\end{equation}
	where the subdifferential is with respect to the first variable. 
	Hence, by \eqref{conjsub}, \eqref{conda} is equivalent to
	\[
	-(w^\ast)^{\mathsf{T}}X\in -\partial\sigma_{\A_2^r}(-V^\ast_2)
	\]
	as well as to
	\[
	-V^\ast_2\in \partial I_{\A_2^{r}}((w^\ast)^{\mathsf{T}}X).
	\]
	In particular, $\partial I_{\A_2^{r}}((w^\ast)^{\mathsf{T}}X)\neq \emptyset$ so that $(w^\ast)^{\mathsf{T}}X\in\A_2^r$, that is, $g_2(w^\ast)\leq r$, and
	\[
	-V^\ast_2\in \partial I_{\A_2^{r}}((w^\ast)^{\mathsf{T}}X)=\N(\A_2^r,(w^\ast)^{\mathsf{T}}X).
	\]
	So
	\[
	\sup_{w\in\R^n\colon g_2(w)\leq r}w^{\mathsf{T}}\E[-V_2^\ast X]=\sup_{w\in\R^n \colon g_2(w)\leq r}\E[-V^\ast_2 w^{\mathsf{T}}X]\leq \sup_{Y\in \A_2^r}\E[-V^\ast_2 Y]=\E[-V^\ast_2(w^\ast)^{\mathsf{T}}X],
	\]
	which implies that the inequality in the middle is indeed an equality. Therefore, the equality of the first and last quantities yields
	\begin{equation}\label{condv2}
		\E[-V_2^\ast X]\in \N(\{g_2\leq r\},w^\ast),
	\end{equation}
	where $\{g_2\leq r\}\coloneqq\cb{w\in\R^n\mid g_2(w)\leq r}$.
	
	Combining the results for the three terms of \eqref{three} for a generic $w\in\R^n_+$, we see that
	\[
	d(r)=\inf_{w\in\R^n_+}\of{\rho_1(w^{\mathsf{T}}X)+I_{\A_2^r}(w^{\mathsf{T}}X)+I_{\W}(w)}=\inf\cb{\rho_1(w^{\mathsf{T}}X)\mid \rho_2(w^{\mathsf{T}}X)\leq r, \ w\in\W}=p(r),
	\]
	establishing the strong duality between $(\P(r))$ and $(\D(r))$.
	
	Note that the inequality constraint in $(\D(r))$ ensures that
	\begin{equation}\label{normal1}
		w^{\mathsf{T}}\E[V^\ast_1 X]\leq w^{\mathsf{T}}(\E[-V_2^\ast X] + y^\ast\1),\quad w\in \W.
	\end{equation}
	Moreover, the complementary slackness condition for this constraint yields
	\begin{equation}\label{normal2}
		(w^\ast)^{\mathsf{T}}\E[V^\ast_1 X]= (w^\ast)^{\mathsf{T}}(\E[-V_2^\ast X] + y^\ast\1).
	\end{equation}
	On the other hand, bringing together \eqref{condy} and \eqref{condv2} gives
	\begin{align*}
		\E[-V_2^\ast X]  + y^\ast\1\in \N(\{g_2\leq r\},w^\ast)+\N(\W,w^\ast)&=\partial I_{\cb{g_2\leq r}}(w^\ast)+\partial I_\W(w^\ast)\\
		&\subseteq \partial (I_{\cb{g_2\leq r}}+ I_\W)(w^\ast)\\
		&=\partial I_{\cb{g_2\leq r}\cap \W}(w^\ast)\\
		&=\N(\cb{g_2\leq r}\cap \W,w^\ast).
	\end{align*}
	In the above calculation, only the passage to the third line is nontrivial and it is justified by the rules of subdifferential calculus; see, for instance, \citet[Theorem~23.8]{rockafellar}. Hence, by \eqref{normal1} and \eqref{normal2}, we have
	\[
	\sigma_{\cb{g_2\leq r}\cap \W}(\E[V^\ast_1 X])\leq \sigma_{\cb{g_2\leq r}\cap \W}(\E[-V_2^\ast X]  + y^\ast\1)=(w^\ast)^{\mathsf{T}}\of{\E[-V_2^\ast X]  + y^\ast\1}=(w^\ast)^{\mathsf{T}}\E[V^\ast_1 X]
	\]
	Hence, by \eqref{normalch}, we conclude that $\E[V^\ast_1 X]\in \N(\cb{g_2\leq r}\cap \W,w^\ast)$, that is, 
	\[
	\E[-V^\ast_1 X]\in -\N(\cb{g_2\leq r}\cap \W,w^\ast).
	\]
	Combining this with \eqref{condy1}, we obtain
	\[
	\E[-V^\ast_1 X]\in \partial g_1(w^\ast)\cap -\N(\cb{g_2\leq r}\cap \W,w^\ast).
	\]
	By \eqref{PRthm}, this implies that $w^\ast$ is a minimizer of $g_1$ over $\cb{g_2\leq r}\cap \W$, that is, $w^\ast$ is an optimal solution of $(\P(r))$. 
\end{proof}

\begin{remark}\label{finrem2}
	It should be noted that we do not work under Assumption~\ref{finiteness} in Theorem~\ref{thm2}, hence the strong duality established by Theorem~\ref{thm1} is not taken for granted; instead we re-establish strong duality in Theorem~\ref{thm2}. Then, the reader might naturally question the need for Theorem~\ref{thm1}. As noted in Remark~\ref{finrem}, in the absence of Assumption~\ref{finiteness}, the arguments in the proof of Proposition~\ref{prop1} (and Theorem~\ref{thm1}) can be seen as a heuristic way to derive the dual problem $(\D(r))$. Theorem~\ref{thm2} provides the formal justification of this heuristic approach without using Sion's minimax theorem.
\end{remark}

\section{Analysis of the problem when the principle risk measure is quasiconvex}\label{quasiconvexcase}

In this section, we remove Assumption~\ref{rho1convex} and study $(\P(r))$ with $\rho_1$ being a quasiconvex risk measure.

Recalling \eqref{recovery}, we may write
\begin{align}
	p(r)&=\inf\cb{\rho_1(w^{\mathsf{T}}X)\mid \rho_2(w^{\mathsf{T}}X)\leq r,\ w\in\W}\notag \\
	&=\inf\cb{t\in\R\mid w^{\mathsf{T}}X\in \A_1^t,\ w^{\mathsf{T}}X\in \A_2^r,\ w\in\W}\notag \\
	&=\inf_{t\in\R}\of{t+f(t,r)}=\inf_{t\in\R: f(t,r)<+\infty}\of{t+f(t,r)},\label{qc-c}
\end{align}
where
\begin{equation}\label{feas}
	f(t,r)\coloneqq \inf_{w\in\W}\of{I_{\A_1^t}(w^{\mathsf{T}}X)+I_{\A_2^r}(w^{\mathsf{T}}X)},\quad t\in\R.
\end{equation}
Note that, for each $t\in\R$, $f(t,r)$ is closely related to the \emph{feasibility problem}
\begin{align*}
	\text{Find }w\in\W\text{ such that }w^{\mathsf{T}}X\in \A_1^t\cap \A_2^r.\tag{$\mathscr{F}^P(t,r)$}
\end{align*}
Indeed, if there exists $w\in\R^n$ solving $(\mathscr{F}^P(t,r))$, then $f(t,r)=0$; otherwise, $f(t,r)=+\infty$. The expression in \eqref{feas} formulates $(\mathscr{F}^P(t,r))$ as an optimization problem whose optimal value is $f(t,r)$; this problem is convex because $\A_1^t,\A_2^r$ are convex sets. Hence, in view of \eqref{qc-c}, the quasiconvex portfolio optimization problem $(\P(r))$ is characterized by a family of convex optimization problems; this is a well-known paradigm in quasiconvex programming as discussed recently in \citet[Section~2.1]{boydnew}.

We introduce an analogue of Assumption~\ref{quasirelintnew} that will be needed in recovering a solution for $(\mathscr{F}^P(t,r))$ below.

\begin{assumption}\label{quasirelintnew2}
	Given $t\in\R$, there exist $V_1,V_2\in L^q_{++}$ such that $\a_1(V_1,t)\in\R$, $\a_2(V_2,r)\in\R$.
\end{assumption}

In the next theorem, following a similar path as in Section~\ref{convexcase} (see Theorem~\ref{thm1}, Theorem~\ref{thm2}), we provide a dual formulation of $f(t,r)$ and a method to calculate a solution for $(\mathscr{F}^P(t,r))$.

\begin{theorem}\label{thm3}
	Let $t\in\R$ and consider the problem
	\begin{align*}
		%		&\text{maximize}\; \;               t-\a_1(V_1,t)-\a_2(V_2,r)-y\tag{$\D(r)$}\\
		&\text{maximize}\; \;               -\a_1(V_1,t)-\a_2(V_2,r)-y\tag{$\mathscr{F}^D(t,r)$}\\
		&\text{subject to}\;\;                   \E[V_1X]+\E[V_2X]\leq y\1\\
		%	& \quad\quad\quad\quad\;\;\;  \tilde{\a}_1(V_1)\in\R \\
		& \quad\quad\quad\quad\;\;\;      V_1, V_2\in L^q_+,\; y\in\R,
	\end{align*}
	\begin{enumerate}[(i)]
		\item Then, $(\mathscr{F}^P(t,r))$ and $(\mathscr{F}^D(t,r))$ have the same optimal value $f(t,r)$.
		\item Suppose that Assumption~\ref{quasirelintnew2} holds for $t$ and there exists an optimal solution $(V^t_1,V^t_2,y^t)$ for $(\mathscr{F}^D(t,r))$. Then, there exists an optimal Lagrange multiplier $w^t\in\R^n$ associated to the inequality constraint of $(\mathscr{F}^D(t,r))$. Moreover, every $w^t\in\R^n$ that is the Lagrange multiplier of the inequality constraint of $(\mathscr{F}^D(t,r))$ at optimality is an optimal solution for $(\mathscr{F}^P(t,r))$.
	\end{enumerate}
\end{theorem}

\begin{proof}
	We first prove (i) under the additional assumption that $\a_1(V_1,t)\in\R, \a_2(V_2,r)$ for all $V_1,V_2\in L^q_+$. Since $\A_1^t, \A_2^r$ are closed convex subsets of $L^p$, similar to the proof of Proposition~\ref{prop1}, we have
	\begin{align*}
		f(t,r)&= \inf_{w\in\W}\of{I_{\A_1^t}(w^{\mathsf{T}}X)+I_{\A_2^r}(w^{\mathsf{T}}X)}\\
		&=\inf_{w\in\W}\of{\sup_{V_1\in L^q_+}\of{\E[-V_1w^{\mathsf{T}}X]-\a_1(V_1,t)}+\sup_{V_2\in L^q_+}\of{\E[-V_2w^{\mathsf{T}}X]-\a_2(V_1,r)}}\\
		&=\inf_{w\in \W}\sup_{V_1\in L^q_+, V_2\in L^q_+}v_{t,r}(w,V_1,V_2),
	\end{align*}
	where
	\[
	v_{t,r}(w,V_1,V_2)\coloneqq \E[-V_1w^{\mathsf{T}}X]+\E[-V_2w^{\mathsf{T}}X]-\a_1(V_1,t)-\a_2(V_2,r).
	\]
	By the properties of support function, it is clear that $(V_1,V_2)\mapsto v_{t,r}(w,V_1,V_2)$ is concave and upper semicontinuous for fixed $w\in\W$. On the other hand, for fixed $(V_1,V_2)\in L^q_+\times L^q_+$, the function $w\mapsto v_{t,r}(w,V_1,V_2)$ is continuous and affine, hence lower semicontinuous convex. Since $\W$ is a convex compact set and $v_{t,r}$ has finite values thanks to our additional assumption, we may apply the standard minimax theorem \citet[Corollary~3.3]{sion} and get
	\begin{align*}
		f(t,r)&=\sup_{V_1\in L^q_+, V_2\in L^q_+}\inf_{w\in\W}v_{t,r}(w,V_1,V_2)\\
		&=\sup_{V_1\in L^q_+, V_2\in L^q_+}\of{\inf\cb{\of{\E[-V_1X]+\E[-V_2X]}^{\mathsf{T}}w\mid \1^{\mathsf{T}}w=1,\ w\in\R^n_+}-\a_1(V_1,t)-\a_2(V_2,r)}.
	\end{align*}
	Note that the inner minimization problem is a finite-dimensional linear optimization problem with nonempty feasible region. Hence, by linear programming duality, we may pass to its dual formulation, which yields
	\begin{align*}
		f(t,r)&=\sup_{V_1\in L^q_+, V_2\in L^q_+}\of{\sup\cb{-y\mid \E[V_1X]+\E[V_2X]\leq y\1,\ y\in\R}-\a_1(V_1,t)-\a_2(V_2,r)}\\
		&=\sup\cb{-\a_1(V_1,t)-\a_2(V_2,r)-y\mid \E[V_1X]+\E[V_2X]\leq y\1,\ V_1,V_2\in L^q_+,\ y\in\R},
	\end{align*}
	which coincides with the optimal value of $(\mathscr{F}^D(t,r))$.
	
	Next, we prove (i) without the additional assumption as well as (ii). Let $\tilde{f}(t,r)$ be the optimal value of $(\mathscr{F}^D(t,r))$. Let $(V_1^t,V_2^t,y^t)\in L^q_+\times L^q_+\times \R$ be an optimal solution for $(\mathscr{F}^D(t,r))$. By Assumption~\ref{quasirelintnew2}, there exist $V_1,V_2\in L^q_{++}$ such that $\a1(V_1,t)\in\R$ and $\a_2(V_2,r)\in\R$. Similar to the proof of Theorem~\ref{thm2}, it follows that Slater's condition holds for $(\mathscr{F}^D(t,r))$; hence, by strong duality for convex optimization, there exists an optimal Lagrange multiplier $w^t\in\R^n$ such that
	\begin{align*}
		&\tilde{f}(t,r)\\
		&=\inf_{w\in\R^n}\sup_{V_1,V_2\in L^q_+, y\in\R}\of{-\a_1(V_1,t)-\a_2(V_2,r)-y-w^{\mathsf{T}}\of{\E[V_1X]+\E[V_2X]-y\1}}\\
		&=\inf_{w\in\R^n}\sup_{V_1,V_2\in L^q_+, y\in\R}\of{-\a_1(V_1,t)-\a_2(V_2,r)-y-\E[V_1w^{\mathsf{T}}X]+\E[V_2w^{\mathsf{T}}X]-yw^{\mathsf{T}}\1}\\
		&=\sup_{V_1,V_2\in L^q_+, y\in\R}\of{-\a_1(V_1,t)-\a_2(V_2,r)-y-\E[V_1(w^t)^{\mathsf{T}}X]-\E[V_2(w^t)^{\mathsf{T}}X]-y(w^t)^{\mathsf{T}}\1}\\
		&=\sup_{V_1\in L^q_+}\of{-\a_1(V_1,t)-\E[V_1(w^t)^{\mathsf{T}}X]}+\sup_{V_2\in L^q_+}\of{-\a_2(V_2,r)-\E[V_2(w^t)^{\mathsf{T}}X]}+\sup_{t\in\R}y\of{1-(w^t)^{\mathsf{T}}\1}.
	\end{align*}
	Then, following similar arguments as in the proof of Theorem~\ref{thm2}, it can be checked that $\tilde{f}(t,r)=f(t,r)$ so that (i) holds without the additional assumption as well. Moreover, it can be checked that $w^t\in\W$, and
	\[
	y^t\1\in\N(\W,w^t),\quad -(w^t)^{\mathsf{T}}X\in \partial \a_1(V_1^t,t),\quad (w^t)^{\mathsf{T}}X\in \partial \a_2(V_2^t,r),
	\]
	which implies that
	\[
	\E[-V_1^t X]\in -\N(\{g_1\leq t \},w^t),\quad \E[-V_2^t X]\in -\N(\{g_2\leq r\},w^t);
	\]
	and finally we obtain
	\[
	\E[-V_1^t X]\in -\N(\{g_1\leq t\}\cap\{g_2\leq r\}\cap\W,w^t).
	\]
	Hence, we conclude that $w^t$ solves the feasibility problem $(\mathscr{F}^P(t,r))$. 
\end{proof}

The next assumption will be useful when devising a method to find an approximately optimal solution for $(\P(r))$.

\begin{assumption}\label{gen}
	It holds $g_1(\1)\in\R$ and $\dom g_1\cap \{g_2\leq r\}\cap \W\neq \emptyset$. In other words, $\rho_1(\1^{\mathsf{T}}X)\in\R$ and there exists $w^0\in\W$ such that $\rho_1((w^0)^{\mathsf{T}}X)\in\R$ and $\rho_2((w^0)^{\mathsf{T}}X)\leq r$.
\end{assumption}

Finally, we discuss a simple method to solve $(\P(r))$ with the help of Theorem~\ref{thm3}. First note that 
\[
p(r)=\inf\cb{g_1(w)\mid g_2(w)\leq r, w\in\W}.
\]
Since $\{g_2\leq r\}\cap \W$ is a convex set and $g_1$ is a lower semicontinuous function, $(\P(r))$ has an optimal solution, that is, there exists $w^\ast\in\W$ such that $g_2(w^\ast)\leq r$ and 
\[
g_1(w^\ast)=p(r).
\] 
Moreover, under Assumption~\ref{gen}, we also have
\[
p(r)=g_1(w^\ast)\leq g_1(w^0)<+\infty.
\]
On the other hand, by the monotonicity of $\rho_1$ and Assumption \ref{gen},
\[
-\infty<g_1(\1)=\rho_1(\1^{\mathsf{T}}X)\leq \rho_1((w^\ast)^{\mathsf{T}}X) = g_1(w^\ast)=p(r)
\]
Hence, $p(r)\in\R$ with finite upper bound $u_1\coloneqq g_1(w^0)$ and finite lower bound $\ell_1\coloneqq g_1(\1)$. Let $\varepsilon>0$. Using these bounds, the well-known bisection algorithm (see \citet[Section~3]{boydnew}) can be employed to find an approximately optimal solution for $(\P(r))$ as follows. At each iteration $k\in\mathbb{N}$, we start with $\ell_k,u_k\in\R$ such that $\ell_k\leq p(r)\leq u_k$ and we let 
\[
t_k\coloneqq \frac{\ell_k+u_k}{2}.
\]
Then, under Assumption~\ref{quasirelintnew}, we solve the feasibility problem $(\mathscr{F}^P(t_k,r))$, that is, we calculate $f(t_k,r)$. If $f(t_k,r)=0$, then we have $\ell_k\leq p(r)\leq t_k$, in which case we proceed to the next iteration using $\ell_{k+1}\coloneqq \ell_k$ and $u_{k+1}\coloneqq t_k$. Otherwise, if $f(t_k,r)=+\infty$, then we have $t_k\leq p(r)\leq u_k$, in which case we proceed to the next iteration using $\ell_{k+1}\coloneqq t_k$ and $u_{k+1}\coloneqq u_k$. We stop this procedure at the first iteration number $K$ for which $u_K-\ell_K\leq \varepsilon$. It can be checked that
\[
K\leq \left\lceil \log_2\of{\frac{g_1(w^0)-g_1(\1)}{\varepsilon}}\right \rceil
\]
so that the algorithm stops in finitely many iterations. Then, we may apply Theorem \ref{thm3} and find an optimal Lagrange multiplier $w^{t_{K}}$, which also solves $(\mathscr{F}^P(t_{K},r))$. Hence, $w^{t_{K}}\in\W$, $g_2(w^{t_{K}})\leq r$ and
\[
p(r)\leq g_1(w^{t_{K}})\leq t_{K}\leq p(r)+\varepsilon,
\]
which shows that $w^{t_{K}}$ is an $\varepsilon$-optimal solution for $(\P(r))$.

\section{Examples}\label{examples}

In this section, we consider some well-known classes of quasiconvex risk measures as special cases of $\rho_1$ and $\rho_2$.

\subsection{Convex risk measures}\label{crm}

Let $\rho$ be a lower semicontinuous convex risk measure on $L^p$ with $\rho(0)\in\R$ and acceptance sets $(\A^t)_{t\in\R}$. Since $\rho$ is translative, its acceptance set $\A\coloneqq\A^0$ at level $0$ determines $\rho$ completely. As a result, the dual representation in Proposition~\ref{dualrep} reduces to a simpler form which we derive here for the convenience of the reader. Let $V\in L^q_+\sm\cb{0}$ and $t\in\R$. Then, by the translativity of $\rho$,
\[
\a(V,t)=\sup_{Y\in L^p\colon \rho(Y)\leq t}\E[-VY]=\sup_{Y\in L^p\colon \rho(Y+t)\leq 0}\E[-VY]=\sup_{Y\in L^p\colon Y+t \in \A}\E[-VY]=\g(V)+t\E[V],
\]
where
\[
\g(V)\coloneqq  \sup_{Y\in\A}\E[-VY].
\]
The function $\g\colon L^q_+\sm\cb{0}\to \bar{\R}$ is called the minimal penalty function of $\rho$ \emph{in the sense of convex risk measures} (not to be confused with the minimal penalty function $\a$ \emph{in the sense of quasiconvex risk measures}); it follows from the definition that $\g$ is convex and lower semicontinuous. Since $V\neq 0$, we have $\E[V]>0$. From this and the above expression for $\a$, it is evident that $t\mapsto\a(V,t)$ is concave (indeed affine) and continuous with
\[
\lim_{t\rightarrow\infty}\a^-(V,t)=\lim_{t\rightarrow\infty}\a(V,t)=+\infty
\]
(see Assumption~\ref{assmpalpha}); and $t\mapsto \a(V,t)$ has a true inverse given by
\[
\b(V,s)=\frac{s-\g(V)}{\E[V]}, \quad s\in\R.
\]
Moreover, for each $a\in\R$, the set
\[
\cb{(V,s)\in L^{q,1}_+\times\R\mid \b(V,s)\geq a}=\cb{(V,s)\in L^{q,1}_+\times\R\mid \g(V)+a\E[V]-s\leq 0}
\]
is closed by the lower semicontinuity of $\g$; therefore, $\b$ is jointly upper semicontinuous on $L^{q,1}_+\times\R$ (see Assumption~\ref{additionalbeta}). On the other hand, we have
\[
\tilde{\a}(V)\coloneqq\inf_{t\in\R}\of{t-\a(V,t)}=\inf_{t\in\R}\of{(1-\E[V])t-\g(V)}=\begin{cases}
	-\g(V)&\text{if }\E[V]=1,\\
	-\infty&\text{else}.
\end{cases}
\]

Now suppose that, in $(\P(r))$, both $\rho_1$ and $\rho_2$ are convex risk measures with respective minimal penalty functions $\g_1$ and $\g_2$ in the sense of convex risk measures. Then, by the above discussion, Assumption~\ref{rho1convex}, Assumption~\ref{additionalbeta} and Assumption~\ref{assmpalpha} hold, and the dual problem $(\D(r))$ can be rewritten as
\begin{align*}
	&\text{maximize}\; \;               -\g_1(V_1)-\g_2(V_2)-r\E[V_2]-y\\
	&\text{subject to}\;\;                   \E[V_1X]+\E[V_2X]\leq y\1 \\
	& \quad\quad\quad\quad\;\;\;      \E[V_1]=1\\
	& \quad\quad\quad\quad\;\;\;      V_1, V_2\in L^q_+,\; y\in\R.
\end{align*}

\begin{example}\label{entropic}
	For this example, we assume that $p=+\infty$ (and $q=1$). For each $j\in\cb{1,2}$, let us suppose that $\rho_j$ is the \emph{entropic risk measure} with risk aversion parameter $r_j>0$ (see, for instance, \citet[Example~4.34]{fs:sf}), that is,
	\[
	\rho_j(Y)=\frac{1}{r_j}\log\E\sqb{e^{-r_j Y}},\quad Y\in L^\infty.
	\]
	In this case, it is well-known that $\g_j$ is the relative entropy function given by
	\[
	\g_j(V)=\frac{1}{r_j}\E\sqb{V\log\of{\frac{V}{\E[V]}}}=\frac{1}{r_j}\of{\E[V\log(V)]-\E[V]\log(\E[V])},\quad V\in L^1_+\sm\cb{0}.
	\]
	Note that Assumption~\ref{quasirelintnew} is satisfied here: by taking $V_1=V_2\equiv 1$, we have $\tilde{\a}_1(V_1)=0\in\R$ and $\a_2(V_2,r)=r\in\R$. 
	Moreover, $(\D(r))$ takes the form
	\begin{align*}
		&\text{maximize}\; \;               -\frac{1}{r_1}\E[V_1\log(V_1)]-\frac{1}{r_2}\E[V_2\log(V_2)]+\frac{1}{r_2}\E[V_2]\log(\E[V_2])-r\E[V_2]-y\\
		&\text{subject to}\;\;                   \E[V_1X]+\E[V_2X]\leq y\1\\
		& \quad\quad\quad\quad\;\;\;      \E[V_1]=1\\
		& \quad\quad\quad\quad\;\;\;      V_1, V_2\in L^1_+,\; y\in\R.
	\end{align*}
	In the case of a finite probability space, this problem can be solved numerically using, for instance, the convex optimization package CVX (see \cite{cvx,cvxpaper}) as it is able to work with the convex function $z\mapsto z\log(z)$ on $\R_+$. As is standard in convex optimization, these packages also provide the value of the Lagrange multiplier $w^\ast$ that corresponds to the inequality constraint in $(\D(r))$ at (approximate) optimality. By Theorem \ref{thm2}, such $w^\ast$ is an (approximately) optimal solution of $(\P(r))$.
\end{example}

\subsection{Certainty equivalents}\label{certaintyequivalent}

Certainty equivalents form an important class of quasiconvex risk measures. We briefly recall their definitions and properties; see \citet[Example~8]{drapeaukupper} for more details. To avoid integrability issues, we assume that $p=+\infty$ (hence $q=1$) in all examples although, in each example, a larger $L^p$ space can be considered depending on the nature of the loss function.

Let $\ell\colon\R\to(-\infty,+\infty]$ be a convex lower semicontinuous increasing function that is differentiable on $\dom \ell\coloneqq \cb{y\in\R\mid \ell(y)<+\infty}$, we call $\ell$ a \emph{loss function}. Let $\rho$ be the \emph{certainty equivalent} corresponding to $\ell$, that is,
\[
\rho(Y)=\ell^{-1}\of{\E[\ell(-Y)]},\quad Y\in L^\infty,
\]
where $\ell^{-1}$ is the left-continuous inverse of $\ell$. The minimal penalty function of $\rho$ is given by
\begin{equation}\label{aexp}
	\a(V,t)=\E\sqb{V h\of{\lambda(V,t)\frac{V}{\E[V]}}},\quad V\in L^1_+\sm\cb{0},\ t\in\R,
\end{equation}
where $h$ is the right-continuous inverse of the derivative $\ell^\prime$ of $\ell$, and $\lambda(V,t)>0$ is a multiplier that is determined by the equation
\begin{equation}\label{lambdaeqn}
	\E\sqb{\ell\of{h\of{\lambda(V,m)\frac{V}{\E[V]}}}}=\ell^{+}(t),
\end{equation}
where $\ell^+$ is the right-continuous version of $\ell$.

Next, we consider some special cases of $\ell$ for which more explicit forms of $\a$ can be obtained.

\begin{example}\label{quadratic}
	Suppose that $\ell$ is the \emph{quadratic loss function} given by
	\[
	\ell(y)=\begin{cases}\frac{1}{2}y^2+y&\text{if }y\geq -1, \\ -\frac12 &\text{else}.\end{cases}
	\] 
	Let $V\in L^1_+\sm\cb{0}$, $t,s\in\R$. After elementary calculations, we may solve \eqref{lambdaeqn} for $\lambda(V,t)$ and use the resulting expression in \eqref{aexp} to get
	\[
	\a(V,t)=\begin{cases}(1+t)\norm{V}_2-\E[V]&\text{if }t> -1,\\ -\E[V]&\text{else},\end{cases} \quad\quad  \b(V,s)=\begin{cases}\frac{s+\E[V]}{\norm{V}_2}-1&\text{if }s>-\E[V],\\ -\infty &\text{else}.\end{cases}
	\]
	In particular, as in Section~\ref{crm}, it is easy to verify that Assumption~\ref{additionalbeta}  and Assumption~\ref{assmpalpha} hold for $\a_1=\a$.
\end{example}

\begin{example}\label{logarithm}
	Suppose that $\ell$ is the \emph{logarithmic loss function} given by
	\[
	\ell(y)=\begin{cases}-\log(-y)&\text{if }y<0, \\ +\infty &\text{else}.\end{cases}
	\] 
	Let $V\in L^1_+\sm\cb{0}$, $t,s\in\R$. Then, after some elementary calculations, we obtain
	\[
	\a(V,t)=\begin{cases}te^{\E\sqb{\log\of{V}}}&\text{if }t<0,\\ +\infty &\text{ else},\end{cases}\quad\quad \b(V,s)=\begin{cases} se^{-\E\sqb{\log\of{V}}}&\text{if }s<0,\\ 0&\text{ else}.\end{cases}
	\]
	It follows that Assumption~\ref{additionalbeta} and Assumption~\ref{assmpalpha} hold for $\a_1=\a$.
\end{example}

%\begin{example}\label{power}
%	Let $a\in (0,1)$ be an exponent and suppose that $\ell$ is the \emph{power loss function} given by
%	\[
%	\ell(y)=\begin{cases}-(-y)^{1-a}&\text{if }y\leq 0, \\ +\infty &\text{else}.\end{cases}
%	\]
%	Let $V\in L^q_+\sm\cb{0}$, $t,s\in\R$. Then,
%	\[
%	\a(V,t)=\begin{cases} t\norm{V}_{\frac{a-1}{a}}&\text{if }t<0,\\ +\infty&\text{else},\end{cases}\quad\quad \b(V,s)=\begin{cases}\frac{s}{\norm{V}_{\frac{a-1}{a}}}&\text{if }s<0,\\ 0& \text{else }.\end{cases}
%	\]
%	and
%	\[
%	\tilde{\a}(V)\coloneqq \inf_{t\in\R}\of{t-\a(V,t)}=\inf_{t<0}\of{t-t\norm{V}_{\frac{a-1}{a}}}=\begin{cases}0 &\text{if }\norm{V}_{\frac{a-1}{a}}\leq 1,\\ -\infty & \text{else}.\end{cases}
%	\]
%	\end{example}

We conclude the paper with an example where $\rho_1,\rho_2$ are assumed to be certainty equivalents whose respective loss functions $\ell_1,\ell_2$ are among the two examples described above, which illustrates a possible special form of $(\mathscr{F}^D(t,r))$. %In this case, it is clear from the expressions of the minimal penalty functions that Assumption \ref{assmpalpha} and Assumption \ref{additionalbeta} are satisfied.

\begin{example}\label{quadlog}
	Suppose that $\ell_1$ is the quadratic loss function in Example~\ref{quadratic} and $\ell_2$ is the logarithmic loss function in Example~\ref{logarithm}. Assume that $r<0$ and take $t\in\R$. If $t\leq -1$, then $(\mathscr{F}^D(t,r))$ can be rewritten as 
	\begin{align*}
		&\text{maximize}\; \;               -\E[V_1]-re^{\E\sqb{\log\of{V_2}}}-y\\
		&\text{subject to}\;\;                   \E[V_1X]+\E[V_2X]\leq y\1\\
		& \quad\quad\quad\quad\;\;\;      V_1, V_2\in L^1_+,\; y\in\R.
	\end{align*}
	If $t\in (-1,0)$, then $(\mathscr{F}^D(t,r))$ becomes
	\begin{align*}
		&\text{maximize}\; \;               (1+t)\norm{V_1}_2-\E[V_1]-re^{\E\sqb{\log\of{V_2}}}-y\\
		&\text{subject to}\;\;                   \E[V_1X]+\E[V_2X]\leq y\1\\
		& \quad\quad\quad\quad\;\;\;      V_1, V_2\in L^1_+,\; y\in\R.
	\end{align*}
	Finally, if $t\geq 0$, then $(\mathscr{F}^D(t,r))$ becomes
	\begin{align*}
		&\text{maximize}\; \;               (1+t)\norm{V_1}_2-\E[V_1]-y\\
		&\text{subject to}\;\;                   \E[V_1X]+\E[V_2X]\leq y\1\\
		& \quad\quad\quad\quad\;\;\;      V_1, V_2\in L^1_+,\; y\in\R.
	\end{align*}
	As noted in Example~\ref{quadratic}, Assumption~\ref{additionalbeta} and Assumption~\ref{assmpalpha} hold in this case. Moreover, Assumption~\ref{quasirelintnew2} holds for every $t\in\R$ trivially since $\a_1$ and $\a_2$ are real-valued. Then, Theorem~\ref{thm3} is applicable and we may apply the procedure described at the end of Section \ref{quasiconvexcase} to find an approximately optimal portfolio for $(\P(r))$.
\end{example}

%Example \ref{quadlog} provides a case in which the risk measures are quasiconvex but not convex, and Theorem \ref{thm2} is still applicable due to the nice structure of $\tilde{\a}_1$. This suggests that Assumption \ref{tildealphaconcave} is not too restrictive and can still work beyond the case where $\rho_1$ is a convex risk measure. We conclude the paper by providing an example in which Assumption \ref{tildealphaconcave} does not hold.
%
%\begin{example}\label{logquad}
%	Suppose that $\ell_1$ is the logarithmic loss function in Example \ref{logarithm} and $\ell_2$ is the quadratic loss function in Example \ref{quadratic}. Assume that $r>-1$ and take $t\in\R$. If $t\leq -1$, then $(\mathscr{F}^D(t,r))$ can be rewritten as 
%	\begin{align*}
%	&\text{maximize}\; \;               -te^{\E[\log(V_1)]}-\E[V_2]-y\\
%	&\text{subject to}\;\;                   \E[V_1X]+\E[V_2X]\leq y\1\\
%	& \quad\quad\quad\quad\;\;\;      V_1, V_2\in L^1_+,\; y\in\R.
%	\end{align*}
%	However, as noted in Example \ref{logarithm}, Assumption \ref{tildealphaconcave} is violated by $\tilde{\a}_1$ and $(\D(r))$ is not a concave maximization problem in this case. Hence, Theorem \ref{thm2} is not applicable for this example although Theorem \ref{thm1} is still applicable, that is, $(\P(r))$ and $(\D(r))$ share the same optimal value.
%\end{example}

%\section*{Acknowledgments}\label{Acknowledgements}

\bibliographystyle{named}

\end{document}